\documentclass[11pt]{article}
\newlength{\actualtopmargin}
\newlength{\actualsidemargin}
\usepackage[tbtags]{amsmath}
\usepackage{amsthm}
\usepackage{amssymb}
\usepackage{amsfonts}
  \theoremstyle{plain}
  \newtheorem{theorem}{Theorem}

  \theoremstyle{definition}
  \newtheorem{definition}[theorem]{Definition}

  \theoremstyle{remark}
  
  \theoremstyle{plain}
  \newtheorem*{theorem*}{Theorem}
  \newtheorem*{lemma*}{Lemma}
  \newtheorem*{corollary*}{Corollary}
  \newtheorem*{proposition*}{Proposition}
  \newtheorem*{claim*}{Claim}

\newenvironment{step}
  {
    \begin{enumerate}

  }
  {\end{enumerate}}

\newenvironment{algorithm*}[1]
  {
    \begin{center}
      \hrulefill\\
      \textbf{#1}
  }
  {
    \vspace{-1\baselineskip}
    \hrulefill
    \end{center}
  }

\newenvironment{protocol*}[1]
  {
    \begin{center}
      \hrulefill\\
      \textbf{#1}
  }
  {
    \vspace{-1\baselineskip}
    \hrulefill
    \end{center}
  }

\newcommand{\bbN}{\mathbb{N}}

\newcommand{\bbZ}{\mathbb{Z}}

\newcommand{\sfB}{\mathsf{B}}

\newcommand{\sfO}{\mathsf{O}}

\newcommand{\sfR}{\mathsf{R}}
\newcommand{\sfS}{\mathsf{S}}

\newcommand{\classfont}{\mathrm}

\newcommand{\BQP}{\classfont{BQP}}

\newcommand{\MA}{\classfont{MA}}

\newcommand{\QMA}{\classfont{QMA}}
\newcommand{\BQNP}{\classfont{BQNP}}

\newcommand{\QCMA}{\classfont{QCMA}}
\newcommand{\MQA}{\classfont{MQA}}

\newcommand{\ZQEXP}{\classfont{ZQEXP}}

\newcommand{\bra}[1]{\langle #1 \vert}

\newcommand{\ket}[1]{\vert #1 \rangle}

\newcommand{\ketbra}[1]{\vert #1 \rangle \langle #1 \vert}

\newcommand{\conjugate}[1]{#1^{\dagger}}

\newcommand{\bignorm}[1]{\bigl\Vert #1 \bigr\Vert}

\newcommand{\abs}[1]{\vert #1 \vert}

\newcommand{\function}[3]{{#1 \colon #2 \to #3}}

\newcommand{\init}{\mathrm{init}}

\newcommand{\illegal}{\mathrm{illegal}}
\newcommand{\acc}{\mathrm{acc}}
\newcommand{\rej}{\mathrm{rej}}
\newcommand{\yes}{\mathrm{yes}}
\newcommand{\no}{\mathrm{no}}
\newcommand{\Ayes}{A_{\yes}}
\newcommand{\Ano}{A_{\no}}

\newcommand{\Natural}{\bbN}
\newcommand{\Integers}{\bbZ}
\newcommand{\Nonnegative}{\Integers^{+}}

\newcommand{\Binary}{{\{ 0, 1 \}}}

\newcommand{\ignore}[1]{}

\begin{document}

\sloppy


\title{\Large
  \textbf{
     Achieving Perfect Completeness in Classical-Witness Quantum Merlin-Arthur Proof Systems
  }\\
}

\author{
  Stephen P. Jordan\footnotemark[1]~~\footnotemark[5]\\
  \and
  Hirotada Kobayashi\footnotemark[2]\\
  \and
  Daniel Nagaj\footnotemark[3]\\
  \and
  Harumichi Nishimura\footnotemark[4]\\
}

\date{}

\maketitle
\thispagestyle{empty}
\pagestyle{plain}
\setcounter{page}{0}

\renewcommand{\thefootnote}{\fnsymbol{footnote}}

\vspace{-5mm}

\begin{center}
{\large
  \footnotemark[1]%
  Applied and Computational Mathematics Division\\
  Information Technology Laboratory\\
  National Institute of Standards and Technology\\
  Gaithersburg, MD, USA\\
  [2.5mm]
  \footnotemark[2]%
  Principles of Informatics Research Division\\
  National Institute of Informatics\\
  Tokyo, Japan\\
  [2.5mm]
  \footnotemark[3]%
  Research Center for Quantum Information\\
  Institute of Physics\\
  Slovak Academy of Sciences\\
  Bratislava, Slovakia\\
  [2.5mm]
  \footnotemark[4]%
  Department of Mathematics and Information Sciences\\
  Graduate School of Science\\
  Osaka Prefecture University\\
  Sakai, Osaka, Japan
}\\
[5mm]
\end{center}

\footnotetext[5]{
  Part of this work was done while at the Institute for Quantum Information,
  California Institute of Technology, Pasadena.
}

\renewcommand{\thefootnote}{\arabic{footnote}}


\begin{abstract}
This paper proves that classical-witness quantum Merlin-Arthur proof systems
can achieve perfect completeness. That is, ${\QCMA = \QCMA_1}$. This
holds under any gate set with which the Hadamard and arbitrary classical
reversible transformations can be exactly implemented, \emph{e.g.},
${\{\textrm{Hadamard, Toffoli, NOT}\}}$. The proof is quantumly
nonrelativizing, and uses a simple but novel quantum technique that
\emph{additively} adjusts the success probability, which may be of
independent interest.
\end{abstract}

\clearpage



\section{Introduction}

$\QCMA$ (also called $\MQA$~\cite{Wat09ECSS,GSU11arXiv}) was first formally\footnote{
  The general notion of
  nondeterminism in quantum computation originates much
  earlier due to Knill~\cite{Kni96TR}.}
defined by Aharonov~and~Naveh~\cite{AhaNav02arXiv}
as the class of decision problems whose solutions
(given as classical bit strings)
can be efficiently verified by a quantum computer.
The letters ``MA'' stand for Merlin-Arthur, as the complexity class is
motivated by the following protocol.
A bit string~$w$ (the purported witness) is provided by 
a computationally unbounded but untrustworthy prover (Merlin) to a
verifier with only polynomial resources (Arthur). 
The verification procedure of the verifier is a polynomial time
\emph{quantum} computation while the witness~$w$ is a \emph{classical} bit string. 
If the verifier is a polynomial-time classical computer, 
then the resulting class is called $\MA$~\cite{Bab85STOC,BM88}. 
If the verifier and witness are both quantum, that is, $w$
is an arbitrary quantum state,
the resulting complexity class is called $\QMA$~\cite{Wat00FOCS}
(originally called $\BQNP$~\cite{Kit99AQIP,KitSheVya02Book}).

The standard way of defining these complexity classes allows
two-sided bounded error: Arthur may wrongly reject each yes-instance
with small probability (completeness error), and may also wrongly
accept each no-instance with small probability (soundness error). If
Arthur never wrongly rejects yes-instances, the system is said to
have \emph{perfect completeness}. The versions of $\QMA$, $\QCMA$, and $\MA$
with perfect completeness are denoted $\QMA_1$, $\QCMA_1$, and $\MA_1$,
respectively.

One of the important open problems in quantum Merlin-Arthur proofs
(both in the $\QCMA$ case and in the $\QMA$ case)
is whether the class defined with two-sided error
equals that with perfect completeness. A proof that
${\QMA = \QMA_1}$ would be particularly interesting, as the problem of
deciding whether a Hamiltonian is frustrated is
$\QMA_1$-complete~\cite{Bra06arXiv}.
Classically, it is known that ${\MA = \MA_1}$ due to Zachos~and~F\"urer~\cite{ZacFur87FSTTCS}
(Goldreich~and~Zuckerman~\cite{GolZuc11LNCS} provided an alternative proof of this).
More generally, it is known that perfect completeness is achievable
in various models of quantum and classical interactive proof systems~\cite{ZacFur87FSTTCS,GolManSip87FOCS,FurGolManSipZac89ACR,BenGolKilWig88STOC,KitWat00STOC,MarWat05CC,KemKobMatVid09CC}.
In contrast, Aaronson~\cite{Aar09QIC} presented
a \emph{quantum} oracle relative to which $\QMA_1$ is a proper
subclass of $\QMA$. This implies that any proof of ${\QMA = \QMA_1}$
must be quantumly nonrelativizing. Aaronson's oracle also separates
$\QCMA$ from $\QCMA_1$ since in fact he showed a quantum oracle 
relative to which $\BQP$ is not contained in the exponential-time analogue of $\QMA_1$. 
He suggested that the result (and the proof) in Ref.~\cite{Aar09QIC} implies that 
any proof of ${\QMA = \QMA_1}$ (and also ${\QCMA = \QCMA_1}$) requires 
some technique of explicitly representing (probability) amplitudes that appear 
in quantum states or evolutions. 

This paper shows that demanding perfect completeness does \emph{not}
weaken the power of QCMA proof systems under a reasonable
assumption on the gate set. Specifically, assuming that 
Hadamard transformations and all classical reversible transformations can be
exactly implemented, ${\QCMA = \QCMA_1}$. To the best of our knowledge,
this is the first ``nontrivial'' example that overcomes a quantum oracle
separation (except quantumly nonrelativizing ``trivial'' containments 
such as ${\BQP \subseteq \ZQEXP}$ as found in Ref.~\cite{Aar09QIC}). 
Our proof of ${\QCMA = \QCMA_1}$ is nonrelativizing 
since our technique also utilizes an explicit representation of amplitudes. 
This suggests that the oracle separation of Ref.~\cite{Aar09QIC} may not be 
an insurmountable barrier to proving ${\QMA = \QMA_1}$. 
We hope that our proof may provide guidance on approaching the
longstanding $\QMA$ versus $\QMA_1$ problem, and on developing
quantumly nonrelativizing techniques in general. It is also
interesting to note that, as a corollary of our result, the solutions
to the known $\QCMA$-complete problems~\cite{Wocjan_Janzing_Beth} can
be verified with perfect completeness. 

Our basic strategy to prove ${\QCMA = \QCMA_1}$ is very simple:
Given any QCMA proof system with two-sided error, one considers letting
Arthur receive a description of the acceptance probability in addition
to the original classical witness.
This allows Arthur to adjust the acceptance probability
by standard exact amplitude amplification~\cite{BraHoyMosTap02ConM,CK98}
or Watrous's quantum rewinding~\cite{Wat09SIComp}.
One obvious problem in this approach is that the original acceptance probability
might not be expressible exactly with polynomially many bits.
This can be overcome by making use of the robustness of
the two-sided error complexity class~$\QCMA$ against the choice of gate set.
Specifically, one can assume without loss of generality that
the verification procedure of the original two-sided error QCMA system
is implemented only with Hadamard, Toffoli, and NOT gates~\cite{Shi_Toffoli,Aharonov_Toffoli}.
This ensures that any possible acceptance probability on input~$x$
in this system is exactly equal to $k/2^{l(\abs{x})}$ for some integer~$k$
and some polynomially bounded, integer-valued function~$l$.

Another problem, which is more difficult to overcome, is that Arthur
may not be able to appropriately adjust the acceptance probability
without error, even if he knows the original acceptance probability.
The standard way to adjust success probability in exact amplitude
amplification is a ``multiplicative'' method that applies some
suitable rotation operator. This rotation depends on the input length,
and cannot be exactly implemented with a fixed finite gate set in general.
We overcome this difficulty by introducing a simple but novel
``additive'' method of adjusting the acceptance probability.
The goal is to have a base procedure whose initial acceptance probability
is exactly $1/2$, which leads to a protocol with perfect completeness
via Watrous's quantum rewinding.
(The choice of quantum rewinding rather than exact amplitude amplification is
just for ease of analysis, and is not essential.)

On input $x$, Arthur receives as a witness a string $w$ and an
integer $k$, written using $l(\abs{x})$ bits, where $w$ is expected to
be the witness he would receive in the original system, and $k$ is
expected such that $k/2^{l(\abs{x})}$ equals 
the acceptance probability $p_{x,w}$ on input $x$ and witness $w$ in
the original system.  If the claimed $k$ is too small relative to the
value computed from the original completeness condition, Arthur
rejects.  Otherwise Arthur performs with equal amplitude the original
verification test and an additional second test, where Arthur
generates a uniform superposition of values from $1$ to
$2^{l(\abs{x})}$ and simply accepts if this value is more than
$k$. Notice that this second test is exactly implementable only with
the Hadamard and classical reversible transformations.  Clearly, the
honest Merlin can prepare some suitable pair $(w,k)$ with which Arthur
accepts with probability $p_{x,w}$ in the original verification test
and with probability ${1 - k/2^{l(\abs{x})} = 1 - p_{x,w}}$ in the
second test.  Hence, this base procedure has its initial success
probability exactly $1/2$ for yes-instances, and one can construct a
system of perfect completeness via quantum rewinding, similar to the
case of quantum multi-prover interactive
proofs~\cite{KemKobMatVid09CC}.  For a dishonest Merlin, any possible $w$
must have a small $p_{x,w}$ value while $k$ must be such that
the value $k/2^{l(\abs{x})}$ is large, and thus, whichever pair
$(w,k)$ is prepared, the initial success probability of the base
procedure must be less than $1/2$, which ensures soundness. To the
best of our knowledge, no such ``additive'' method of amplitude
adjustment has appeared in the literature previously, and we believe
it may have other applications in quantum complexity theory.


\section{Preliminaries}
\label{Section: Preliminaries}

We assume the reader is familiar with the quantum formalism, in
particular the quantum circuit model (see
Refs.~\cite{NieChu00Book,KitSheVya02Book}, for instance). Throughout
this paper, let $\Natural$ and $\Nonnegative$ denote the sets of
positive and nonnegative integers, respectively.  A
function~$\function{f}{\Nonnegative}{\Natural}$ is \emph{polynomially
  bounded} if there exists a polynomial-time deterministic Turing
machine that outputs ${1^{f(n)}}$ on input~$1^n$.  A
function~$\function{f}{\Nonnegative}{[0,1]}$ is \emph{negligible} if,
for every polynomially bounded
function~$\function{g}{\Nonnegative}{\Natural}$, it holds that ${f(n)
  < 1/g(n)}$ for all but finitely many values of $n$.

For a quantum register~$\sfR$, let $\ket{0}_{\sfR}$ denote the state
in which all the qubits in $\sfR$ are in state~$\ket{0}$.
In this paper, all Hilbert spaces have dimension a power of two.


\paragraph{Polynomial-Time Uniformly Generated Families of Quantum Circuits}

Following conventions,
we define quantum Merlin-Arthur proof systems
in terms of quantum circuits.
In particular, we use the following notion of
polynomial-time uniformly generated families of quantum circuits.

A family~${\{ Q_x \}}$ of quantum circuits is
\emph{polynomial-time uniformly generated}
if there exists a deterministic procedure
that, on every input~$x$, outputs a description of $Q_x$
and runs in time polynomial in $\abs{x}$.
It is assumed that the circuits in such a family are composed of gates
in some reasonable, universal, finite set of quantum gates.
Furthermore, it is assumed that the number of gates in any circuit
is not more than the length of the description of that circuit.
Therefore $Q_x$ must have size polynomial in $\abs{x}$.
For convenience,
we may identify a circuit~$Q_x$ with the unitary operator it induces.

Throughout this paper, we assume a gate set with which the Hadamard
and any classical reversible transformations can be exactly
implemented.  Note that this assumption is satisfied by many standard
gate sets such as the Shor basis~\cite{Sho96FOCS} consisting of the
Hadamard, controlled-$i$-phase-shift, and Toffoli gates, and the one
consisting of the Hadamard and Toffoli gates~\cite{Shi_Toffoli,Aharonov_Toffoli}.
Hence we believe that our condition is
reasonable and not restrictive. For concreteness, we may assume the
specific gate set ${\{\textrm{Hadamard, Toffoli, NOT}\}}$ for both the original $\QCMA$ verifer and our
corresponding $\QCMA_1$ verifier. Note that,
although ${\{\textrm{Hadamard, Toffoli}\}}$ is computationally
universal~\cite{Shi_Toffoli,Aharonov_Toffoli} given a supply of both
$\ket{0}$ and $\ket{1}$ ancilla qubits, we include the NOT gate
because we assume the verifier receives all qubits initialized to
$\ket{0}$. The witness string $w$ is hardcoded into the verifier
circuit $V_{x,w}$ by initial NOT gates acting on each witness bit
whose value should be 1.

Since non-unitary and unitary quantum circuits
are equivalent in computational power~\cite{AhaKitNis98STOC},
it is sufficient to treat only unitary quantum circuits,
which justifies the above definition. However, we describe our
verification procedure using intermediate projective measurements
in the computational basis and unitary operations conditioned on the
outcome of the measurements. If we wished, we could defer all of the
measurements of the verification procedure to the end of the
computation, along the lines described on page 186
of~Ref.~\cite{NieChu00Book}. 

More specifically, one sees from Figure~\ref{Figure: QCMA_1 protocol}
in the next section
that our verification procedure involves Boolean-outcome measurements
at Steps~1,~3.1,~and~3.5. Let $b_1$, $b_{3.1}$, and $b_{3.5}$ denote
the outcomes of these measurements. Final acceptance occurs if
\begin{equation}
\label{bool}
\lnot b_1 \land (b_{3.1} \lor b_{3.5})
\end{equation}
evaluates to true. This can be determined unitarily using
Toffoli gates, as they can perform universal classical
computation. One may worry that exact unitary implementation of the
conditional operations would require the addition of
conditional-Hadamard to our gate set. However, the operations
in Steps~2~and~3 which, for conceptual clarity, we describe as being performed
only under certain measurement outcomes, can in fact be performed
unconditionally without affecting final acceptance. By
construction, the formula~(\ref{bool}) simply ignores the outcomes of
these steps in the cases that they are irrelevant.

\paragraph{Classical-Witness Quantum Merlin-Arthur Proof Systems}

This paper discusses the power of quantum Merlin-Arthur proof systems
where Merlin sends a classical witness to Arthur,
which we call \emph{QCMA proof systems}.

Formally, the class~$\QCMA(c,s)$ of problems
having such systems with completeness~$c$ and soundness~$s$
is defined as follows.
For generality, we use promise problems~\cite{ESY84} rather than languages
when defining complexity classes.

\begin{definition}
Given functions~$\function{c, s}{\Nonnegative}{[0,1]}$,
a promise problem~${A = (\Ayes, \Ano)}$ is in ${\QCMA(c,s)}$
iff there exists a polynomially bounded function~$\function{m}{\Nonnegative}{\Natural}$
and a polynomial-time quantum verifier~$V$,
who is a polynomial-time uniformly generated family
of quantum circuits~$\{V_{x,w}\}_{x\in\Binary^*, w\in\Binary^{m(\abs{x})}}$,
such that, for every input~$x$:
\begin{description}
\item[\textnormal{(Completeness)}]
if ${x \in \Ayes}$,
there exists a witness~${w \in \Binary^{m(\abs{x})}}$
with which $V$ accepts $x$
(i.e., the measurement on the output qubit of $V_{x,w}$ results in $\ket{1}$)
with probability at least ${c(\abs{x})}$,
\item[\textnormal{(Soundness)}]
if ${x \in \Ano}$,
for any witness~${w' \in \Binary^{m(\abs{x})}}$ given,
$V$ accepts $x$
with probability at most ${s(\abs{x})}$.
\end{description}
\label{Definition: QCMA(c,s)}
\end{definition}

The complexity class~$\QCMA$ is defined as follows.

\begin{definition}
A promise problem~${A = (\Ayes, \Ano)}$ is in $\QCMA$
iff $A$ is in ${\QCMA(1- \varepsilon, \varepsilon)}$
for some negligible function~$\function{\varepsilon}{\Nonnegative}{[0,1]}$.
\label{Definition: QCMA}
\end{definition}

Similarly, the class $\QCMA_1$ is defined as follows.

\begin{definition}
A promise problem~${A = (\Ayes, \Ano)}$ is in $\QCMA_1$
iff $A$ is in ${\QCMA(1, \varepsilon)}$
for some negligible function~$\function{\varepsilon}{\Nonnegative}{[0,1]}$.
\label{Definition: QCMA_1}
\end{definition}

Note that ${\QCMA = \QCMA(2/3, 1/3)}$ and ${\QCMA_1 = \QCMA(1, 1/2)}$,
since the gap between completeness and soundness can be amplified exponentially
by repeating the verification procedure.


\section{Result}
\label{Section: Result}

Now we show that any QCMA proof system with two-sided error
can be converted into
another QCMA proof system with perfect completeness.

\begin{theorem}
${\QCMA = \QCMA_1}$.
\label{Theorem: QCMA = QCMA_1}
\end{theorem}

In fact, we show a more general theorem stated
below. Theorem~\ref{Theorem: QCMA = QCMA_1} is an immediate corollary. 

\begin{theorem}
For any polynomial-time computable function~$\function{c}{\Nonnegative}{[0,1]}$
and any function~$\function{s}{\Nonnegative}{[0,1]}$ satisfying
${c - s \geq 1/q}$ for some polynomially bounded function~$\function{q}{\Nonnegative}{\Natural}$,
\footnote{
  Actually, it is sufficient for our proof that
  ${c - s \geq 1/2^q}$ for some polynomially bounded function~$\function{q}{\Nonnegative}{\Natural}$.
}
\[
\QCMA(c,s) \subseteq \QCMA(1,s'),
\]
where ${s'=\frac{1}{2}\bigl(1-(c-s)\bigr)\bigl(1+(1+c-s)^2\bigr) < 1}$.
\label{Theorem: QCMA(c,s) is in QCMA(1,s')}
\end{theorem}

By taking ${c=2/3}$ and ${s=1/3}$,
Theorem~\ref{Theorem: QCMA(c,s) is in QCMA(1,s')} implies
${\QCMA(2/3,1/3) \subseteq \QCMA(1,25/27)}$,
which is sufficient to obtain Theorem~\ref{Theorem: QCMA = QCMA_1}.

The rest of this section is devoted to the proof of
Theorem~\ref{Theorem: QCMA(c,s) is in QCMA(1,s')}.

\begin{proof}[Proof of Theorem~\ref{Theorem: QCMA(c,s) is in QCMA(1,s')}]
Let ${A = (\Ayes, \Ano)}$ be in ${\QCMA(c,s)}$
and let $V$ be the verifier of the corresponding QCMA system.
Consider the quantum circuit $V_{x,w}$ of $V$
when the input is $x$ and the received witness is $w$ of ${m(\abs{x})}$ bits,
for some polynomially bounded function~$\function{m}{\Nonnegative}{\Natural}$.
Without loss of generality, by Refs.~\cite{Aharonov_Toffoli, Shi_Toffoli},
one can assume that $V_{x,w}$ consists of
only the Hadamard, Toffoli, and NOT gates,
and the output is obtained by measuring the designated
output qubit in the computational basis.  Therefore, the acceptance
probability of $V_{x,w}$ is exactly expressible as
$k_{x,w}/2^{l(\abs{x})}$ for some integer~$k_{x,w}$ in $\{0,1,\ldots,2^{l(\abs{x})}\}$,
where $\function{l}{\Nonnegative}{\Natural}$ is a polynomially bounded
function such that ${l(\abs{x})}$ denotes the size of the circuit~$V_{x,w}$.

We construct a new verifier~$W$ assuring that $A$ is in ${\QCMA(1,s')}$.
Let $\sfR$ be the quantum register consisting of all the qubits used by $V_{w,x}$.
The verifier $W$ uses three more quantum registers~$\sfB$,~$\sfO$,~and~$\sfS$
in addition to $\sfR$,
where $\sfB$ and $\sfO$ are single-qubit registers,
and $\sfS$ is a quantum register of ${l(\abs{x})}$ qubits.
All the qubits in these four registers are initialized to $\ket{0}$.
The qubit in $\sfO$ is designated as the output qubit in the constructed system.
As a witness, $W$ receives binary strings~$w$~and~$k$,
where $w$ is expected to be the witness $V$ would receive in the original system,
and $k$ is an ${l(\abs{x})}$-bit string
that identifies a positive integer in $\{1,\ldots,2^{l(\abs{x})}\}$
that is expected to be $k_{x,w}$.
Here notice that we are considering a natural one-to-one correspondence
between ${l(\abs{x})}$-bit strings and integers from $1$ to $2^{l(\abs{x})}$
(rather than from $0$ to ${2^{l(\abs{x})} - 1}$),
for $k_{x,w}$ cannot be zero but can be $2^{l(\abs{x})}$ in the yes-instance case.
$W$ immediately rejects if $k$ viewed as an integer is less than
${c(\abs{x}) \cdot 2^{l(\abs{x})}}$.

Then $W$ applies the Hadamard transformations
over all qubits in registers~$\sfB$~and~$\sfS$,
and applies the original verification circuit~$V_{x,w}$ over the qubits in $\sfR$.
$W$ accepts either when $\sfB$ contains $0$
and the content of $\sfR$ would result in acceptance in the original system,
or when $\sfB$ contains $1$ and the content of $\sfS$ viewed as an integer
expressed by an ${l(\abs{x})}$-bit string is greater than $k$
(the qubit in $\sfO$, which is the output qubit of the constructed system,
is flipped to $\ket{1}$ in these two cases). 
Otherwise $W$ continues by performing the quantum rewinding procedure.
The precise description of the protocol of $W$ is given in Figure~\ref{Figure: QCMA_1 protocol}.
It is easy to see that this protocol is exactly implementable using
only the Hadamard and classical reversible transformations
(the protocol includes intermediate measurement, which can be
postponed until the very end of the protocol
via standard technique that only uses classical reversible transformations).

\begin{figure}[t!]
\begin{algorithm*}{Verifier's Protocol for Achieving Perfect Completeness}
\begin{step}
\item
  Receive an ${m(\abs{x})}$-bit string~$w$
  and an integer $k$ in $\{1,\ldots,2^{l(\abs{x})}\}$
  expressed by an ${l(\abs{x})}$-bit string as witness.
  Reject if ${k/2^{l(\abs{x})}<c(\abs{x})}$. 
\item
  Perform the following unitary transformation~$Q$
  over the qubits in ${(\sfB, \sfO, \sfR, \sfS)}$.
  \begin{step}
  \item
    Apply the Hadamard transformations to all the qubits in $\sfB$ and $\sfS$, 
    and apply $V_{x,w}$ to the qubits in $\sfR$.
  \item
    Apply the bit-flip to the qubit in $\sfO$
    either when $\sfB$ contains $0$
    and the content of $\sfR$ would result in acceptance in the original system,
    or when $\sfB$ contains $1$
    and the content of $\sfS$
    viewed as an integer in $\{1,\ldots,2^{l(\abs{x})}\}$
    is greater than $k$.
  \end{step}
\item
  Do the following steps (quantum rewinding):
  \begin{step}
    \item
      Accept if $\sfO$ contains $1$, and continue otherwise.
    \item
      Invert Step 2.
      That is, apply $\conjugate{Q}$ to ${(\sfB, \sfO, \sfR, \sfS)}$.
    \item
      Perform the phase-flip (i.e., multiply $-1$ in phase)
      if all the qubits in ${(\sfB, \sfO, \sfR, \sfS)}$ are in state~$\ket{0}$.
    \item
      Perform the same operations as in Step 2.
      That is, apply $Q$ to ${(\sfB, \sfO, \sfR, \sfS)}$.
    \item
      Accept if $\sfO$ contains $1$, and reject otherwise.
  \end{step}
\end{step}
\end{algorithm*}
\caption{Verifier's protocol for achieving perfect completeness. Note
  that the conditional phase flip can be exactly achieved using
  ${\{\textrm{Hadmard, Toffoli, NOT}\}}$ by preparing
  an ancilla qubit $\ket{0}$ into the state $\frac{1}{\sqrt{2}} \left(
  \ket{0} - \ket{1} \right)$ by the application of a NOT followed by a
  Hadamard, and then performing a conditional NOT on that qubit.}
\label{Figure: QCMA_1 protocol}
\end{figure}

Now we analyze the protocol.
Our analysis is similar to the proof of Lemma~3.3 in Ref.~\cite{KemKobMatVid09CC}
(which is based on the ideas in Refs.~\cite{MarWat05CC,Wat09SIComp}).

Let $\Pi_{\init}$ be the projection onto the all-zero state
(i.e., the state in which all the qubits in ${(\sfB, \sfO, \sfR, \sfS)}$
are in state~$\ket{0}$, which is denoted by $\ket{0}_{(\sfB, \sfO, \sfR, \sfS)}$)
and let $\Pi_{\acc}$ be the projection onto states
in which the qubit in $\sfO$ is in state~$\ket{1}$.
Let $Q$ be the unitary transformation induced by the actions in Step~2.
Conditioned on $W$ \emph{not} rejecting in Step~1,
the probability of being accepted in Step~3.1 can be written as
${p_{x,w,k}=\bignorm{\Pi_{\acc}Q \ket{0}_{(\sfB, \sfO, \sfR, \sfS)}}^2}$.
This implies that the matrix
${M=\Pi_{\init} \conjugate{Q} \Pi_{\acc} Q \Pi_{\init}}$
is expressed as
\[
M
=
\Pi_{\init} \conjugate{Q} \Pi_{\acc} Q \Pi_{\init}
=
p_{x,w,k}\ketbra{0}_{(\sfB, \sfO, \sfR, \sfS)}
=
p_{x,w,k}\Pi_{\init},
\]
since 
\begin{align*}
\Pi_{\init} \conjugate{Q} \Pi_{\acc} Q \Pi_{\init}
&=
\ket{0}_{(\sfB, \sfO, \sfR, \sfS)}
\bigl(
  \bra{0}_{(\sfB, \sfO, \sfR, \sfS)} \conjugate{Q} \Pi_{\acc} Q \ket{0}_{(\sfB, \sfO, \sfR, \sfS)}
\bigr)
\bra{0}_{(\sfB, \sfO, \sfR, \sfS)}
\\
&=
\bignorm{
  \Pi_{\acc} Q \ket{0}_{(\sfB, \sfO, \sfR, \sfS)}
}^2
\ketbra{0}_{(\sfB, \sfO, \sfR, \sfS)}.
\end{align*}

Define the unnormalized states
$\ket{\phi_0}$, $\ket{\phi_1}$, $\ket{\psi_0}$, and $\ket{\psi_1}$
by 
\begin{align*}
  \ket{\phi_0}
  &=
  \Pi_{\acc} Q \ket{0}_{(\sfB, \sfO, \sfR, \sfS)},
  &
  \ket{\phi_1}
  &=
  \Pi_{\rej} Q \ket{0}_{(\sfB, \sfO, \sfR, \sfS)},
  &
  \ket{\psi_0}
  &=
  \Pi_{\init} \conjugate{Q} \ket{\phi_0},
  &
  \ket{\psi_1}
  &=
  \Pi_{\illegal} \conjugate{Q} \ket{\phi_0},
\end{align*}
where $\Pi_{\illegal}$ is the projection onto states orthogonal to
$\ket{0}_{(\sfB, \sfO, \sfR, \sfS)}$
and
$\Pi_{\rej}$ is that onto states in which the qubit in $\sfO$ is in state~$\ket{0}$.

First, we analyze the acceptance probability of $W$
when the claimed $k$ satisfies ${k \geq c(\abs{x}) \cdot 2^{l(\abs{x})}}$
(i.e., when $W$ does not reject in Step~1).
Clearly, $W$ accepts in Step~3.1 with probability~$p_{x,w,k}$.
We analyze the probability of being accepted in Step~3.5.
For this purpose, 
it suffices to follow the changes of the unnormalized state~
${
  \ket{\phi_1}
  =
  \Pi_{\rej} Q \ket{0}_{(\sfB, \sfO, \sfR, \sfS)}
}$
during the protocol when $W$ continues in Step~3.1.
Since 
\[
\ket{\psi_0}
=
\Pi_{\init} \conjugate{Q} \Pi_{\acc} Q \ket{0}_{(\sfB, \sfO, \sfR, \sfS)}
=
\Pi_{\init} \conjugate{Q} \Pi_{\acc} Q \Pi_{\init} \ket{0}_{(\sfB, \sfO, \sfR, \sfS)}
=
M \ket{0}_{(\sfB, \sfO, \sfR, \sfS)}
=
p_{x,w,k} \ket{0}_{(\sfB, \sfO, \sfR, \sfS)},
\]
the state\footnote{The norm of this state is
  $\sqrt{1-p_{x,w,k}}$. Alternatively, one could carry out the
  analysis with a conventional normalized state, in which case
  one would obtain the \emph{conditional} probability of accepting in
  Step 3.5 given that Step 3.1 does not accept. This conditional probability
  is one in the case where ${p_{x,w,k}=1/2}$.} just after Step~3.2 is
\begin{align*}
\conjugate{Q} \Pi_{\rej} Q \ket{0}_{(\sfB, \sfO, \sfR, \sfS)}
&=
\ket{0}_{(\sfB, \sfO, \sfR, \sfS)} - \conjugate{Q} \ket{\phi_0}
\\
&=
\frac{1}{p_{x,w,k}} \ket{\psi_0} - \bigl( \ket{\psi_0} + \ket{\psi_1} \bigr)
\\
&=
\frac{1-p_{x,w,k}}{p_{x,w,k}} \ket{\psi_0} - \ket{\psi_1}.
\end{align*}
As ${\Pi_{\init} \ket{\psi_0} = \ket{\psi_0}}$ and ${\Pi_{\init} \ket{\psi_1} = 0}$, 
the controlled phase-flip in Step~3.3 changes the state to
\begin{align*}
-\frac{1-p_{x,w,k}}{p_{x,w,k}} \ket{\psi_0} - \ket{\psi_1}
&=
-\frac{1-2p_{x,w,k}}{p_{x,w,k}} \ket{\psi_0} - \bigl( \ket{\psi_0}+\ket{\psi_1} \bigr)
\\
&= -(1-2p_{x,w,k}) \ket{0}_{(\sfB, \sfO, \sfR, \sfS)} - \conjugate{Q} \ket{\phi_0}.
\end{align*}
Using ${Q \ket{0}_{(\sfB, \sfO, \sfR, \sfS)} = \ket{\phi_0} + \ket{\phi_1}}$,
one can see that the state just after Step~3.4 is
\[
-(1-2p_{x,w,k}) Q \ket{0}_{(\sfB, \sfO, \sfR, \sfS)} - \ket{\phi_0}
=
-(2-2p_{x,w,k}) \ket{\phi_0} - (1-2p_{x,w,k}) \ket{\phi_1}.
\]
Thus, the probability of being accepted in Step~3.5 is
\[
(2-2p_{x,w,k})^2 \bignorm{\ket{\phi_0}}^2
=
4 p_{x,w,k} (1-p_{x,w,k})^2.
\] 
Hence, the acceptance probability~$p_{\acc}$ of $W$
when the claimed $k$ satisfies ${k \geq c(\abs{x}) \cdot 2^{l(\abs{x})}}$
is given by
\[
p_{\acc}
=
p_{x,w,k} + 4 p_{x,w,k} (1-p_{x,w,k})^2.
\]

Now we calculate ${p_{x,w,k}=\bignorm{\Pi_{\acc}Q \ket{0}_{(\sfB, \sfO, \sfR, \sfS)}}^2}$.
Notice that  
\begin{align*}
Q \ket{0}_{(\sfB, \sfO, \sfR, \sfS)}
=
\frac{1}{\sqrt{2^{l(\abs{x})+1}}}
\Bigl(
&
  \ket{0}_{\sfB} \ket{0}_{\sfO}
  \bigl(
    \ket{0} \ket{\chi_0}
  \bigr)_{\sfR}
  \sum_{z \in \{1, \ldots, 2^{l(\abs{x})}\}} \ket{z}_{\sfS}
\\
&
  +
  \ket{0}_{\sfB} \ket{1}_{\sfO}
  \bigl(
    \ket{1}\ket{\chi_1}
  \bigr)_{\sfR}
  \sum_{z \in \{1, \ldots, 2^{l(\abs{x})}\}} \ket{z}_{\sfS}
\\
&
  +
  \ket{1}_{\sfB} \ket{0}_{\sfO}
  \bigl(
    \ket{0} \ket{\chi_0} + \ket{1} \ket{\chi_1}
  \bigr)_{\sfR}
  \sum_{z \in \{1, \ldots, k\}} \ket{z}_{\sfS}
\\
&
  +
  \ket{1}_{\sfB} \ket{1}_{\sfO}
  \bigl(
    \ket{0} \ket{\chi_0} + \ket{1} \ket{\chi_1}
  \bigr)_{\sfR}
  \sum_{z \in \{k+1, \ldots, 2^{l(\abs{x})}\}} \ket{z}_{\sfS}
\Bigr),
\end{align*}
where we denote the state~$V_{x,w}\ket{0}_{\sfR}$
just before the final measurement in the original system by
${
  \ket{0}\ket{\chi_0} + \ket{1} \ket{\chi_1},
}$
assuming that the first qubit in $\sfR$ was the output qubit in the original system.
Since ${\bignorm{\ket{\chi_1}}^2 = k_{x,w}/2^{l(\abs{x})}}$, one can see that 
\[
p_{x,w,k}
=
\frac{1}{2} \cdot \frac{k_{x,w}}{2^{l(\abs{x})}}
+
\frac{1}{2} \cdot \frac{2^{l(\abs{x})}-k}{2^{l(\abs{x})}}
=
\frac{1}{2} - \frac{1}{2^{l(\abs{x})+1}}(k-k_{x,w}).
\]

Now we are ready to verify the completeness and soundness of the constructed protocol.

For the completeness, 
one can take $w$ to be any string that achieves
${k_{x,w}/2^{l(\abs{x})} \geq c(\abs{x})}$
(recall that
the acceptance probability of $V_{x,w}$ is $k_{x,w}/2^{l(\abs{x})}$),
and $k$ to be $k_{x,w}$ for the chosen $w$. 
Then clearly $W$ does not reject in Step~3.1
and we have ${p_{x,w,k}=1/2}$, which implies ${p_{\acc}=1}$. 
Thus, $W$ accepts $x$ with certainty.

For the soundness, note that one has only to consider the case
where ${k/2^{l(\abs{x})} \geq c(\abs{x})}$,
as otherwise $W$ rejects with certainty in Step~3.1.
Since for any $w$ the acceptance probability~$k_{x,w}/2^{l(\abs{x})}$ of $V_{x,w}$
is at most ${s(\abs{x})}$, 
\[
p_{x,w,k}
\leq
\frac{1}{2}
-
\frac{1}{2^{l(\abs{x})+1}}
  \bigl(
    c(\abs{x}) \cdot 2^{l(\abs{x})} - s(\abs{x}) \cdot 2^{l(\abs{x})}
  \bigr)
=
\frac{1}{2} - \frac{c(\abs{x})-s(\abs{x})}{2}.
\]
Noting that the function~${f(p) = p + 4p(1-p)^2}$
is monotone increasing over $[0,1/2]$ and ${f(1/2)=1}$,
one can see that $p_{\acc}$ is at most
\begin{align*}
f \biggl(
    \frac{1}{2}-\frac{c(\abs{x})-s(\abs{x})}{2}
  \biggr)
&=
\frac{1}{2}
-
\frac{c(\abs{x})-s(\abs{x})}{2}
+
4 \biggl(
    \frac{1}{2} - \frac{c(\abs{x})-s(\abs{x})}{2}
  \biggr)
  \biggl(
    \frac{1}{2} + \frac{c(\abs{x})-s(\abs{x})}{2}
  \biggr)^2
\\
&=
\frac{1}{2}
\Bigl(
  1 - \bigl( c(\abs{x})-s(\abs{x}) \bigr)
\Bigr)
\Bigl( 1 + \bigl(1 + c(\abs{x})-s(\abs{x}) \bigr)^2 \Bigr),
\end{align*}
which is smaller than ${f(1/2) = 1}$.
Thus the soundness follows, which completes the proof.
\end{proof}


\section{Concluding Remarks}
\label{Section: Concluding Remarks}

This paper has proved that ${\QCMA=\QCMA_1}$ holds
under any gate set with which the Hadamard and arbitrary classical
reversible transformations can be exactly implemented.
As already mentioned, this result is not quantumly relativizing.
It should be noted, however, that it is classically relativizing
(i.e., ${\QCMA^A=\QCMA_1^A}$ for any classical oracle $A$).
Here we assume the standard model of classical oracles
in computational complexity theory,
in particular that the answer of $A$ for any query is deterministic.
This fact can be easily seen:
for any specific choice of $A$, the acceptance probability of
the verifier can still be represented in the form of $k/2^{l(|x|)}$
since $A$ is deterministic.

A natural question to ask is whether one can extend our argument
to the $\QMA$ case to show that ${\QMA=\QMA_1}$.
There seem to be at least two obstacles for this.
First, the maximum acceptance probability of the verifier
(even in the honest Merlin case)
cannot be expressed with a polynomial number of bits in general.
This is because the maximum acceptance probability in the QMA system
corresponds to the largest eigenvalue of a certain appropriate matrix,
which might only be describable as a zero of some polynomial with exponentially
many terms.
Second, even if one knew its probability as an algebraic number,
it is not easy to boost the probability to one via amplitude amplification or quantum rewinding
-- without an explicit description of the initial state we do not see a way to perform 
a perfect reflection about the initial state, which seems to be necessary
(see, \emph{e.g.}, Ref.~\cite{NWZ09}).


\subsection*{Acknowledgements}

The authors thank Jake~Taylor, Michele~Mosca, and Pawel~Wocjan for
useful discussions, and an anonymous reviewer for helpful comments on
the earlier version of this paper.  Part of this work was performed
while SJ was at the Institute for Quantum Computation at Caltech.  He
gratefully acknowledges the support he received from the Sherman
Fairchild Foundation and NSF~grant~PHY-0803371 and thanks the Slovak
Academy of Sciences for hospitality.  HK is partially supported by the
Grant-in-Aid for Scientific Research~(B)~No.~21300002 of the Japan
Society for the Promotion of Science.  DN gratefully acknowledges
support from the Slovak Research and Development Agency under the
contract No.~LPP-0430-09, from the project APVV-0646-10, and European
project Q-ESSENCE.  HN is partially supported by the Grant-in-Aid for
Scientific Research~(A)~Nos.~21244007 and 23246071 of the Japan
Society for the Promotion of Science and the Grant-in-Aid for Young
Scientists~(B)~No.~22700014 of the Ministry of Education, Culture,
Sports, Science and Technology in Japan.

\clearpage


\bibliographystyle{alpha}
\bibliography{JKNN12v13arXiv}

\end{document}